\documentclass[11pt]{article}
\usepackage[font=small,format=plain,labelfont=bf,up]{caption}
\usepackage[a4paper,nohead,ignorefoot, 
twoside,scale=0.85,bindingoffset=1.25cm]{geometry}
\usepackage{amssymb,amsfonts,amsmath,amsthm}
\usepackage[]{graphicx}
\usepackage{color}
\graphicspath{{figures/}}
\newcommand{\texorpdfstring}[2]{#1}   
\newcommand{\url}[1]{#1} 
\usepackage{hyperref}%
\definecolor{gray}{rgb}{0.2,0.2,.2}
\hypersetup{%
  unicode=true,          
  colorlinks=true,       
  linkcolor=gray,          
  citecolor=gray      
}
\DeclareMathOperator{\sinc}{\mathrm{sinc}}
\DeclareMathOperator{\id}{\mathrm{id}}
\newcommand{\even}{\mathrm{even}}
\def\longrightharpoonup{
\relbar\joinrel\joinrel\relbar\joinrel\joinrel\relbar\joinrel\joinrel\relbar\joinrel\joinrel\relbar\joinrel\joinrel\relbar\joinrel\joinrel\rightharpoonup}
\newcommand{\xrightharpoonup}[1]{\stackrel{#1}{\longrightharpoonup}}
\newcommand{\tdots}{{...}}%
\newcommand{\upidx}[1]{{\at{#1}}}
\newcommand{\D}{\displaystyle}
\newcommand{\ul}[1]{\underline{#1}}
\newcommand{\bigpar}{\par\quad\newline\noindent}
\newcommand{\dint}[1]{\,\mathrm{d}#1}
\newcommand{\iu}{\mathtt{i}}
\newcommand{\mhexp}[1]{{{\mathtt{e}}^{#1}}}
\newcommand{\fspace}[1]{{\mathsf{#1}}}
\newcommand{\fspaceL}{\fspace{L}}

\newcommand{\fspaceW}{\fspace{W}}

\newcommand{\Rset}{{\mathbb{R}}}

\newcommand{\Nset}{{\mathbb{N}}}

\newcommand{\ocinterval}[2]{(#1,\,#2]}%
\newcommand{\ccinterval}[2]{[#1,\,#2]}%

\newcommand{\DO}[1]{{O\at{#1}}}

\newcommand{\skp}[2]{{\left\langle{#1},\,{#2}\right\rangle}}

\newcommand{\pair}[2]{{\left({#1},\,{#2}\right)}}

\newcommand{\at}[1]{{\left({#1}\right)}}

\newcommand{\bat}[1]{{\big(#1\big)}}
\newcommand{\Bat}[1]{{\Big(#1\Big)}}

\newcommand{\ato}[1]{{\left[{#1}\right]}}

\newcommand{\norm}[1]{\left\|{#1}\right\|}
\newcommand{\nnorm}[1]{\|{#1}\|}
\newcommand{\bnorm}[1]{\big\|{#1}\big\|}
\newcommand{\Bnorm}[1]{\Big\|{#1}\Big\|}

\newcommand{\abs}[1]{\left|{#1}\right|}
\newcommand{\babs}[1]{\big|{#1}\big|}
\newcommand{\Babs}[1]{\Big|{#1}\Big|}

\newcommand{\al}{{\alpha}}
\newcommand{\be}{{\beta}}
\newcommand{\ga}{{\gamma}}

\newcommand{\eps}{{\varepsilon}}

\newcommand{\ka}{{\kappa}}
\newcommand{\la}{{\lambda}}

\newcommand{\om}{{\omega}}

\newcommand{\calA}{\mathcal{A}}
\newcommand{\calB}{\mathcal{B}}

\newcommand{\calF}{\mathcal{F}}

\newcommand{\calI}{\mathcal{I}}

\newcommand{\calL}{\mathcal{L}}
\newcommand{\calM}{\mathcal{M}}
\newcommand{\calN}{\mathcal{N}}

\newcommand{\calP}{\mathcal{P}}
\newcommand{\calQ}{\mathcal{Q}}

\theoremstyle{plain}
\newtheorem{theorem}{Theorem}[]
\newtheorem{corollary}      [theorem]{Corollary}

\newtheorem{lemma}         [theorem]{Lemma}

\newtheorem*{result*}{Main result}
\newtheorem*{problem*}{Open problems}

\newtheorem{assumption} [theorem]{Assumption}
\usepackage{float}
\begin{document}


\title{KdV waves in atomic chains with nonlocal interactions}

\date{\today}

\author{%
  Michael Herrmann\footnote{University of M\"unster, Institute for Applied Mathematics, {\tt michael.herrmann@uni-muenster.de}} \and Alice Mikikits-Leitner\footnote{Technical University Munich, Center for Mathematics, {\tt mikikits@ma.tum.de}}
 }

\maketitle


 \begin{abstract} 
We consider atomic chains with nonlocal particle interactions and prove the existence of near-sonic solitary waves. Both our result and the general proof strategy are reminiscent of the seminal paper by Friesecke and Pego on the KdV limit of chains with nearest neighbor interactions but differ in the following two aspects: First, we allow for a wider class of atomic systems and must hence replace the distance profile by the velocity profile. Second, in the asymptotic analysis we avoid a detailed Fourier pole characterization of the nonlocal integral operators and employ the contraction mapping principle to solve the final fixed point problem.
\end{abstract}


 \quad\newline\noindent%
 \begin{minipage}[t]{0.15\textwidth}%
Keywords: 
 \end{minipage}%
 \begin{minipage}[t]{0.8\textwidth}%
\emph{asymptotic analysis}, \emph{KdV limit of lattice waves},\\
\emph{Hamiltonian lattices with nonlocal coupling} 
 \end{minipage}%
 \medskip
 \newline\noindent
 \begin{minipage}[t]{0.15\textwidth}%
   MSC (2010): %
 \end{minipage}%
 \begin{minipage}[t]{0.8\textwidth}%
37K60,  
37K40,  
74H10  	
 \end{minipage}%
%
%
\setcounter{tocdepth}{3}
\setcounter{secnumdepth}{3}{\scriptsize{\tableofcontents}}
%
%
%
\section{Introduction}\label{sect:intro}
%
%
Since the pioneering paper \cite{ZK65}, the so-called KdV limit of atomic chains with nearest neighbor interactions -- often called Fermi-Pasta-Ulam or FPU-type chains -- has attracted a lot of interest in both the physics and the mathematics community, see \cite{FM14} for a recent overview. The key observation is that in the limiting case of long-wave-length data with small amplitudes the dynamics of the nonlinear lattice system is governed by the Korteweg-de Vries (KdV) equation, which is a completely integrable PDE and hence well understood. For rigorous results concerning initial value problems 
we refer to \cite{SW00} and to \cite{CCPG12,GMWZ14} for similar result in chains with periodically varying masses.
\par
Of particular interest are the existence of KdV-like solitary waves and their stability with respect to the FPU dynamics. Both problems have been investigated by Friesecke and Pego in the seminal four-paper series \cite{FP99,FP02,FP04a,FP04b}, see also \cite{HW13} for simplifications in the stability proof and \cite{FM14} concerning the existence of periodic KdV-type waves. The more general cases of two or finitely many solitary waves have been studied in \cite{HW08,HW09} and \cite{Miz11,Miz13}, respectively. In this paper we generalize the existence result from \cite{FP99} and prove that chains with interactions between further than nearest-neighbors also admit KdV-type solitary waves. The corresponding stability problem is beyond the scope of this paper and left for future research.
%
%
\subsection{Setting of the problem}
%
We consider an infinite chain of identical particles which interact with up to $M$ neighbors on both sides. Assuming unit mass, the equations of motion are therefore given by
\begin{align}
\label{LawOfMotion}
\ddot{u}_j=\sum_{m=1}^M\Phi_m^\prime\at{u_{j+m}-u_j}-
\Phi_m^\prime\at{u_{j}-u_{j-m}}\,,
\end{align}
where $u_j\at{t}$ denotes the position of particle $j$ at time $t$. Moreover, the potential $\Phi_1$ describes the interactions between nearest-neighbors, $\Phi_2$ between the next-to-nearest-neighbors, and so on. 
\par
A traveling wave is an exact solution to \eqref{LawOfMotion} which satisfies
\begin{align}
\label{eqn:TW.ansatz}
u_j\at{t}=r_*j + v_*t+\eps\,U_\eps \at{x}\,,\qquad x :=\eps{j}-\eps{c_\eps t}\,,
\end{align}
where the parameters $r_*$ and $v_*$ denote the prescribed background strain and background velocity,  respectively. Moreover, $\eps>0$ is an additional scaling parameter which will be identified below and becomes small in the KdV limit. A direct computation reveals that the wave speed $c_\eps$ as well as the rescaled wave profile $U_\eps$ must solve the rescaled traveling wave equation
\begin{align} 
\label{eq:scaledfpu1}
\eps^3c_\eps^2\,U^{\prime\prime}_\eps=\sum_{m=1}^M m\eps\nabla_{-m\eps}\Phi^\prime_m\at{mr_*+ m\eps^2\nabla_{+m\eps}U_\eps}\,,
\end{align}
where the discrete differential operators are defined by
\begin{align} 
\label{eq:diffoperators}
\at{\nabla_{+m\eps}Y}\at{x}:=\frac{Y\at{x+m\eps}-Y\at{x}}{m\eps}\,,\qquad
\at{\nabla_{-m\eps}Y}\at{x}:=\frac{Y\at{x}-Y\at{x-m\eps}}{m\eps}\,.
\end{align}
Note that $v_*$ does not appear in \eqref{eq:scaledfpu1} due to the Galilean invariance of the problem and that the solution set is invariant under the addition of constants to $U_\eps$. It is therefore natural to interpret \eqref{eq:scaledfpu1} as an equation for the rescaled velocity profile $W_\eps:=U^\prime_\eps$; the corresponding distance or strain profile $\nabla_{+\eps}U_\eps$ can then be computed by convoluting $W_\eps$ with the rescaled indicator function of an interval, see formula \eqref{eq:scaledfpu1a.px1} below.
\par
For $M=1$ and fixed $\eps>0$ there exist -- depending on the properties of $\Phi_1$ -- many different types of traveling waves with periodic, homoclinic, heteroclinic, or even more complex shape of the profile $W_\eps$, see for instance \cite{Her10,HR10, HMSZ13} and references therein. In the limit $\eps\to0$, however, the most fundamental waves are periodic and solitary waves, for which $W_\eps$ is either periodic or decays to $0$
as $x\to\pm\infty$. 
\par
In this paper we suppose $r_*=0$ -- this condition can always be ensured by elementary transformations -- and split off both the linear and the quadratic terms from the force functions $\Phi^\prime_m$. This reads
\begin{align}
\label{eqn:ForceTerms}
\Phi_m^\prime\at{r}=\al_m r + \beta_m r^2 + \Psi_m^\prime\at{r}\,,\qquad 
\Psi_m^{\prime}\at{r}=\DO{r^3}\,,\qquad m=1\tdots M
\end{align}
or, equivalently, $\Phi_m\at{r}=\tfrac12 \al_mr^2 + \tfrac13 \beta_m r^3 + \Psi_m\at{r}$ with $\Psi_m\at{r}=\DO{r^4}$.  In order to keep the presentation as simple as possible, we restrict our considerations to solitary waves -- the case of periodic profiles can be studied along the same lines -- and rely on the following standing assumption.
\begin{assumption}[properties of the interaction potentials]
\label{MainAssumption}
For all $m=1\tdots M$, the coefficients $\alpha_m$ and $\beta_m$ are positive. Moreover, $\Psi_m^\prime$ is continuously differentiable with $\Psi_m^\prime\at{0}=0$ and
\begin{align}
\label{MainAssumption.Eqn1}
\abs{\Psi_m^{\prime\prime}\at{r}}\leq \ga_m r^2
\end{align}
for some constants $\ga_m$ and all $r$ with $\abs{r}\leq 1$.
\end{assumption}
Note that the usual requirements for $M=1$ are $\al_1>0$ and $\be_1\neq0$ but the case $\be_1<0$ can be traced back to the case $\be_1>0$ by a simple reflection argument with respect to the strain variable $r$. Below we discuss possible generalizations of Assumption \ref{MainAssumption} including cases in which the coefficients come with different signs.
\begin{figure}[t!]%
\centering{%
\includegraphics[width=0.6\textwidth]{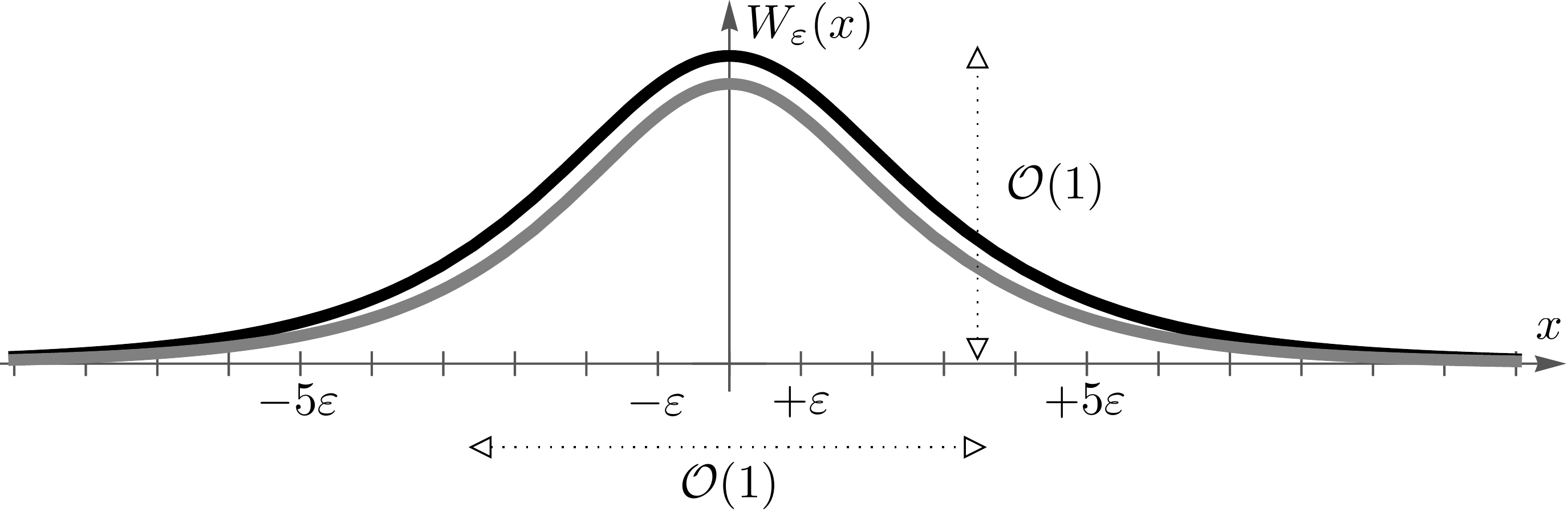}%
}%
\caption{
Sketch of the rescaled velocity profile $W_\eps$ for $\eps>0$ (black) and $\eps=0$ (gray) as function of the rescaled phase variable $x$. The grid with spacing $\eps>0$ describes the rescaled particle index $\eps j$ while the dashed arrows indicate the height and the width of the pulse $W_\eps$. The rescaled distance profile $\calA_\eps W_\eps$ has a similar shape.
}%
\label{Fig0}%
\end{figure}%
%
%
%
\subsection{Overview on the main result and the proof strategy}
%
The overall strategy for proving the existence of KdV-type solitary waves in the lattice system \eqref{LawOfMotion} is similar to the approach in \cite{FP99} but many aspects are different due to the nonlocal coupling. In particular, we base our analysis on the velocity profile 
\begin{align}
\label{Eqn:Def.Vel.prof}
W_\eps:=U_\eps^\prime
\end{align}
instead of the distance profile $\nabla_{\eps}U_\eps$, deviate in the justification of the key asymptotic estimates, and solve the final nonlinear corrector problem by the Banach fixed-point theorem. A more detailed comparison is given throughout the paper.
\par
As for the classical case $M=1$, we prescribe a wave speed $c_\eps$ that is slightly larger than the sound speed $c_0$ and construct  profile functions that satisfy \eqref{eq:scaledfpu1} and decay for $x\to\pm\infty$. More precisely, we set
\begin{align} 
\label{eq:speedscaling}%
c_\eps^2 := c_0^2+\eps^2\, ,\qquad 
c_0^2:=\sum_{m=1}^M \al_m m^2>0\,,
\end{align}
i.e., the small parameter $\eps>0$ quantifies the supersonicity of the wave. Note that the subsonic case $c_\eps<c_0$ is also interesting but not related to solitary waves, see discussions at the end of \S\ref{sect:prelim} and the end of \S\ref{sect:proof}. 
\par
The asymptotic analysis from \S\ref{sect:prelim} reveals that the limiting  problem as $\eps\to0$ is the nonlinear ODE 
\begin{align}
\label{Eqn:WaveODE}
W_0^{\prime\prime} = d_1 W_0- d_2 W_0^2\,,
\end{align}
where the positive constants $d_1$ and $d_2$ depend explicitly on the coefficient $\al_m$ and $\be_m$, see formula \eqref{Eqn:LeadingOrderProblem.2x} below. This equation admits a homoclinic solution, which is unique up to shifts (see \S\ref{sect:proof.1}) and provides via $w\pair{t}{x}=W_0\at{x- t}$ a solitary wave to the KdV equation
\begin{align*}
d_1\,\partial_t w + d_2\,\partial_x w^2 + \partial_x^3 w=0\,.
\end{align*}
For $\eps>0$ we start with the ansatz
\begin{align}
\label{eqn.def.corrector}
W_\eps = W_0+\eps^2 V_\eps \in\fspaceL^2_\even\at\Rset
\end{align}
and derive in \S\ref{sect:proof} the fixed point equation
\begin{align}
\label{Eqn:FixedPoint}
V_\eps = \calF_\eps\ato{V_\eps}
\end{align}
for the corrector $V_\eps$, where the operator $\calF_\eps$ is introduced in \eqref{Thm:FixedPoints.Eqn1}. The definition of $\calF_\eps$ requires to invert a linear operator $\calL_\eps$, which is defined in \eqref{Eqn:DefLandM} and admits a singular limit as $\eps\to0$. The linear leading order operator $\calL_0$ stems from the linearization of \eqref{Eqn:WaveODE} around the KDV wave $W_0$ and can be inverted on the space $\fspaceL^2_\even\at\Rset$ but not on $\fspaceL^2\at\Rset$ due to the shift invariance of the problem. The first technical issue in our perturbative existence proof is to show that this invertibility property persists for small $\eps>0$, see Theorem \ref{Lem:InvertibilityOfLeps}. The second one is to guarantee that $\calF_\eps$ is contractive on some ball in $\fspaceL^2_\even\at\Rset$, see Theorem \ref{Thm:FixedPoints}. Our main findings are illustrated in Figure~\ref{Fig0} and can be summarized as follows, see also Corollary~\ref{Cor:Summary}.
\begin{result*} 
For any sufficiently small $\eps>0$ there exists a unique even and nonnegative solution $W_\eps$ to the rescaled traveling wave equation \eqref{eq:scaledfpu1} with \eqref{eq:speedscaling} such that
\begin{align*}
\norm{W_\eps-W_0}_2+\norm{W_\eps-W_0}_\infty\leq C \eps^2
\end{align*}
holds for some constant $C$ independent of $\eps$, where $W_0$ is the unique even solution to \eqref{Eqn:WaveODE}.
\end{result*}
The asymptotic analysis presented below can -- for the price of more notational and technical effort -- be applied to a wider class of chains. Specifically, we expect that the following generalizations are feasible:
\begin{enumerate}
\item
We can allow for $M=\infty$ provided that the coefficients
$\al_m$, $\be_m$ and $\ga_m$ decay sufficiently fast with respect to $m$ (say, exponentially).
\item 
Some of the coefficients $\alpha_m$ and $\be_m$ might even be negative. In this case, however, one has to ensure that the contributions from the negative coefficients are compensated by those from the positive ones. A first natural condition is
\begin{align*}
\sum_{m=1}^M \al_m m^2>0
\end{align*}
which ensures that uniform states are stable under small amplitude perturbations and that the sound speed $c_0$ from \eqref{eq:speedscaling} is positive. A further minimal requirement is
\begin{align*}
\sum_{m=1}^M \al_m m^4>0\,,\qquad
\sum_{m=1}^M \be_m m^3\neq 0
\end{align*}
because otherwise the leading order problem -- see  \eqref{Eqn:WaveODE} and \eqref{Eqn:LeadingOrderProblem.2x} below -- degenerates and does not admit exponentially decaying homoclinic orbits.
\item 
The non-quadratic contributions to the forces might be less regular in the sense of
\begin{align*}
\abs{\Psi^{\prime\prime}\at{r}}\leq \ga_m \abs{r}^{1+\ka_m}
\end{align*}
for some constants $\ga_m$ and exponents $0<\kappa_m<1$.
\end{enumerate}
The paper is organized as follows: In \S\ref{sect:prelim} we introduce a family of convolution operators and reformulate \eqref{eq:scaledfpu1} as an eigenvalue problem for $W_\eps$. Afterwards we provide singular asymptotic expansions for a linear auxiliary operator $\calB_\eps$, which is defined in \eqref{Eqn:OperatorBeps} and plays a prominent role in our method. \S\ref{sect:proof} is devoted to the proof of the existence theorem. We first study the leading order problem in \S\ref{sect:proof.1} and show afterwards in \S\ref{sect:proof.2} that the linear operator $\calL_\eps$ is invertible. In \S\ref{sect:proof.3} we finally employ Banach's contraction mapping principle to construct solutions $V_\eps$ to the nonlinear fixed problem \eqref{Eqn:FixedPoint} and conclude with a brief outlook.  A list of all important symbols is given in the appendix.
%
\section{Preliminaries and linear integral operators}\label{sect:prelim}
%
%
In this section we reformulate the nonlinear advance-delay-differential equation \eqref{eq:scaledfpu1} as an integral equation and provide asymptotic estimates for the arising linear integral operators.
%
\subsection{Reformulation in terms of convolution operators}
%
%
For any $\eta>0$, we define the operator $\calA_{\eta}$ by
\begin{align} 
\label{eq:intoperator}
\at{\calA_{\eta}Y}\at{x}:=\frac{1}{\eta}\int_{x-\eta/2}^{x+\eta/2}Y\at{\xi}\dint\xi
\end{align}
and regard \eqref{eq:scaledfpu1} as an equation for the rescaled velocity profile \eqref{Eqn:Def.Vel.prof}. Notice that $\calA_\eta$ can be viewed as the convolution with the rescaled indicator function of the interval $\ccinterval{-\eta/2}{+\eta/2}$. 
\begin{lemma}[reformulation as nonlinear eigenvalue problem]
Suppose that $W_\eps=U^\prime_\eps$ belongs to $\fspaceL^2\at\Rset$. Then, the nonlinear eigenvalue problem
\begin{align} 
\label{eq:scaledfpu1a}%
\eps^2 c_\eps^2 W_\eps  = \sum_{m=1}^M m \calA_{m\eps}\Phi_m^\prime\at{m\eps^2 \calA_{m\eps} W_\eps}
\end{align}
is equivalent to the traveling wave equation \eqref{eq:scaledfpu1}. 
\end{lemma}
\begin{proof}
The operators defined in \eqref{eq:diffoperators} and \eqref{eq:intoperator} satisfy
\begin{align}
\label{eq:scaledfpu1a.px1}
\at{\nabla_{\pm m\eps}U_\eps}\at{x}= \bat{
\calA_{m\eps}U_\eps}^\prime\at{x\pm\tfrac12m\eps}=\bat{
\calA_{m\eps}W_\eps}\at{x\pm\tfrac12m\eps}\,,
\end{align}
so \eqref{eq:scaledfpu1} follows from \eqref{eq:scaledfpu1a} after differentiation with respect to $x$ and defining $U_\eps$ as the primitive of $W_\eps$. In order to derive \eqref{eq:scaledfpu1a} from \eqref{eq:scaledfpu1}, we first notice that $W_\eps\in\fspaceL^2\at\Rset$ implies $\calA_{m\eps}W_\eps\in\fspaceW^{1,2}\at\Rset$ (cf. Corollary \ref{Cor:RegularityOperatorA} below) and hence  $\at{\calA_{m\eps}W_\eps}\at{x}$ to $0$ as $x\to\pm\infty$. Afterwards we integrate \eqref{eq:scaledfpu1} with respect to $x$ and eliminate the constant of integration by means of the decay condition at infinity.
\end{proof}
In the case $M=1$, we can derive from \eqref{eq:scaledfpu1a} the identity
\begin{align*}
\eps^2 c_\eps^2 \calA_\eps W_\eps  =  \calA_{\eps}^2\Phi_1^\prime\at{\eps^2 \calA_{\eps}  W_\eps}\,,
\end{align*}
which is the equation for the distance profile $\calA_{\eps}W_\eps$ and has been studied in \cite{FP99} (see equation (2.7) there for the function $\phi=\calA_\eps W_\eps$). For $M>1$, however, we have to work with the velocity profile $W_\eps$ since for a general function $W$ it is not possible to express $\calA_{m\eps}W$ for $m>1$ in terms of $\calA_\eps W$.
\bigpar
We next summarize important properties of the convolution operators defined in \eqref{eq:intoperator}.
\begin{lemma}[properties of $\calA_\eta$]
\label{Lem:PropertiesOperatorA}
For each $\eta>0$, the integral operator $\calA_\eta$ has the following properties:
\begin{enumerate}
\item 
For any $W\in\fspaceL^2\at\Rset$, we have 
$\calA_\eta W\in\fspaceL^2\cap\fspaceL^\infty\at\Rset$ with
\begin{align}
\label{Lem:PropertiesOperatorA.Eqn1}
\norm{\calA_\eta W}_\infty \leq \eta^{-1/2}\norm{W}_2\,,\qquad  \norm{\calA_\eta W}_2 \leq \norm{W}_2\,.
\end{align}
Moreover, $\calA_\eta W$ admits a weak derivative with $\norm{\at{\calA_\eta W}^\prime}_2\leq 2\eta^{-1}\norm{W}_2$.
\item 
For any $W\in\fspaceL^\infty\at\Rset$, we have $\norm{\calA_\eta W}_\infty\leq \norm{W}_\infty$. 
\item 
$\calA_\eta$ respects the even-odd parity, the nonnegativity, and the unimodality of functions. The latter means monotonicity for both negative and positive arguments.
\item 
$\calA_\eta$ diagonalizes in Fourier space and corresponds to
the symbol function
\begin{align}
\label{Eqn:SymbolFct}
a_\eta\at{k}=\sinc\at{\eta k/2}
\end{align}  
with $\sinc\at{z}:=\sin\at{z}/z$.
\item 
$\calA_\eta$ is self-adjoint in the $\fspaceL^2$-sense.
\end{enumerate}
\end{lemma}
\begin{proof}
All assertions follow immediately from the definition of $\calA_\eta$; see \cite{Her10} for the details.
\end{proof}
\begin{corollary}[regularity of $\calA_\eta W$]
\label{Cor:RegularityOperatorA} 
$W\in\fspaceL^2\at\Rset$ implies 
$\calA_\eta W\in\fspace{W}^{1,2}\at{\Rset}\subset\fspace{BC}\at{\Rset}$ and hence
$\at{\calA_\eta W}\at{x}\to0$ as $x\to\pm\infty$.
\end{corollary}
%
%
\subsection{Asymptotic analysis for  the convolution operators \texorpdfstring{$\calA_\eta$}{}}
%
The symbol function $a_\eta$ from \eqref{Eqn:SymbolFct} is analytic with respect to $z=\eta k/2$ and in view of 
\begin{align*}
\sinc\at{z}=\sum_{j=0}^\infty \frac{\at{-1}^j z^{2j}}{\at{2j+1}!}
\end{align*}
we readily verify
\begin{align*}
\calA_{\eta} \mhexp{\iu k{x}}=\sinc\at{\eta k/2} \mhexp{\iu k{x}}=
\sum_{j=0}^{\infty} (-1)^j \frac{\eta^{2j}k^{2j}\mhexp{\iu k{x}}}{2^{2j}\at{2j+1}!}=
\sum_{j=0}^\infty\frac{\eta^{2j}\partial_{x}^{2j}\mhexp{\iu k{x}}}{2^{2j}\at{2j+1}!}\,.
\end{align*}
The integral operator \eqref{eq:intoperator} therefore admits the \emph{formal} expansion
\begin{align} 
\label{Eqn:AsymptoticsA}
\calA_\eta=\sum_{j=0}^\infty\frac{\eta^{2j}\partial_{x}^{2j}}{2^{2j}\at{2j+1}!}\qquad
\text{and hence}\qquad
\calA_{m\eps}=\id +\eps^2\frac{m^2}{24}\partial_{x}^2+\DO{\eps^4}\,,
\end{align}
which reveals that $\calA_{m\eps}$ should be regarded as a \emph{singular perturbation} of the identity operator $id$. This singular nature complicates the analysis because the error terms in \eqref{Eqn:AsymptoticsA} can only be bounded in terms of higher derivatives.
\par%
One key observation for dealing with the limit $\eps\to0$ is -- roughly speaking -- that the resolvent-type operator
\begin{align}
\notag
\at{\id + \ka \frac{\id - \calA_{m\eps}^2}{\eps^2}}^{-1}
\end{align}
is well-defined and almost compact as long as $\ka>0$. It thus exhibits nice regularizing properties which allows us 
to compensate bad terms stemming from the expansion \eqref{Eqn:AsymptoticsA}. The same idea has been employed in \cite{FP99} in the context of the distance profile $\calA_\eps W$, showing  that 
the Yosida-type operator
\begin{align*}
\at{\id +\ka \frac{\id - \calA_{\eps}^2}{\eps^2}}^{-1} \calA_{\eps}^2
\end{align*}
behaves nicely since the corresponding Fourier symbol
\begin{align*}
\frac{\eps^2 a_\eps^2\at{k}}{\eps^2+\ka\at{1-a_\eps^2\at{k}}}
\end{align*}
is well-defined and bounded by  $C/\at{1+\eps^2k^2}$, cf. \cite[Corollary 3.4.]{FP99}. Before we establish a related but weaker result in next subsection, we derive explicit error bounds for the singular expansion of $\calA_{m\eps}$.
\begin{figure}[t!]%
\centering{%
\includegraphics[width=0.8\textwidth]{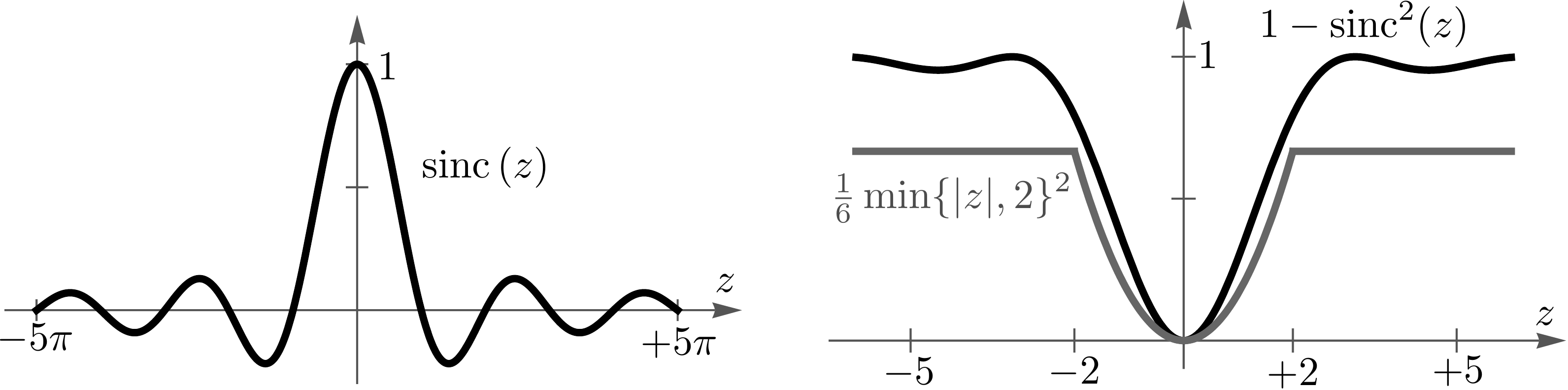}%
}%
\caption{ %
\emph{Left panel:} Graph of the $\sinc$ function
$z\mapsto \sin\at{z}/z$. \emph{Right panel} Lower bound for $1-\sinc^2$ as used in the proof of Lemma \ref{Lem:InversOfB}.
}%
\label{Fig1}%
\end{figure}%
\begin{lemma}[small-parameter asymptotics of $\calA_{\eta}$ ]
\label{Lem:LimitOperatorA}
There exists a constant $C$, which does not depend on $\eta$,
such that the estimates
\begin{align}
\label{Lem:LimitOperatorA.EqnA}
\norm{\calA_{\eta}W-W}_2 \leq C\eta^2 \norm{W^{\prime\prime}}_2\,,\qquad
\norm{\calA_{\eta}W-W}_\infty\leq C\eta^2 \norm{W^{\prime\prime}}_\infty
\end{align}
and
\begin{align}
\label{Lem:LimitOperatorA.EqnB}
\norm{\calA_{\eta}W-W-\frac{\eta^2}{24}W^{\prime\prime}}_2 \leq C\eta^4 \norm{W^{\prime\prime\prime\prime}}_2\,,\qquad
\norm{\calA_{\eta}W-W-\frac{\eta^2}{24}
W^{\prime\prime}}_\infty\leq C\eta^4 \norm{W^{\prime\prime\prime\prime}}_\infty
\end{align}
hold for any sufficiently regular $W$. In particular, we have
\begin{align}
\label{Lem:LimitOperatorA.Eqn0}
\calA_{\eta} W\quad\xrightarrow{\;\eta\to0\;}\quad W\qquad\text{strongly}\quad\text{in}\quad \fspaceL^2\at\Rset
\end{align} 
for any $W\in\fspaceL^2\at\Rset$.
\end{lemma}
\begin{proof} 
\emph{\ul{$\fspaceL^\infty$-estimates}}: 
For any $W\in\fspace{W}^{4,\infty}\at\Rset$, the weak variant of  Taylor's expansion theorem implies
\begin{align*}
\Babs{P\pair{x}{\xi}}\leq
\norm{W^{\prime\prime\prime\prime}}_\infty
\frac{\at{x-\xi}^4}{24}
\end{align*}
for almost all $x,\xi\in\Rset$, where
\begin{align*}
P\pair{x}{\xi}:=
W\at\xi-W\at{x}-W^\prime\at{x}\bat{x-\xi}-
\tfrac{1}{2}W^{\prime\prime}\at{x}\bat{x-\xi}^2-
\tfrac{1}{6}W^{\prime\prime\prime}\at{x}\bat{x-\xi}^3\,.
\end{align*}
Integrating $P\pair{x}{\xi}$ with respect to $\xi\in\ccinterval{x-\eta/2}{x+\eta/2}$ we therefore get
\begin{align*}
\abs{\eta \calA_\eta W\at{x}-\eta W\at{x}-\frac{\eta^3}{24}W^{\prime\prime}\at{x}}&=\abs{\int_{x-\eta/2}^{x+\eta/2}P_\eta\pair{x}{\xi}\dint\xi}\\&\leq
\frac{\norm{W^{\prime\prime\prime\prime}}_\infty}{24}\int_{x-\eta/2}^{x+\eta/2}\at{x-\xi}^4\dint\xi=C\norm{W^{\prime\prime\prime\prime}}_\infty\eta^5\,,
\end{align*}
and \eqref{Lem:LimitOperatorA.EqnB}$_2$ follows immediately. The derivation of \eqref{Lem:LimitOperatorA.EqnA}$_2$ is similar.
\par
\emph{\ul{$\fspaceL^2$-estimates}}:
Now let $W\in\fspaceW^{4,2}\at\Rset$ be arbitrary.
By Parseval's Theorem -- and employing that
$\abs{1-\sinc\at{z}- z^2/6}\leq C z^4$ holds for some constant $C$ and all $z\in\Rset$ -- we find
\begin{align*}
\norm{\calA_\eta W -W-\frac{\eta^2}{24}W^{\prime\prime}}_2^2 &=
\norm{\widehat{W}-\widehat{\calA_\eta W}+\frac{\eta^2}{24}\widehat{W^{\prime\prime}}}_2^2 
\\&=\int_\Rset \at{1-\sinc\at{\eta k/2}-\frac{\eta^2k^2}{24}}^2\abs{\widehat{W}\at{k}}^2\dint{k}
\\&\leq  C\eta^8\int_\Rset \abs{k^4\widehat{W}\at{k}}^2\dint{k}=C\eta^8\norm{W^{\prime\prime\prime\prime}}_2^2\,,
\end{align*} 
and this implies \eqref{Lem:LimitOperatorA.EqnB}$_1$. The estimate \eqref{Lem:LimitOperatorA.EqnA}$_1$
can by proven analogously since we have $\abs{1-\sinc\at{z}}\leq z^2/6$ for all $z\in\Rset$. 
\par
\emph{\ul{Final argument}}: 
Let $W\in\fspaceL^2\at\Rset$ be arbitrary but fixed. Since $\calA_\eta$ is self-adjoint, see Lemma \ref{Lem:PropertiesOperatorA}, and in view of \eqref{Lem:LimitOperatorA.EqnA}  we readily demonstrate
\begin{align}
\label{Lem:LimitOperatorA.PEqn0}
\calA_{\eta} W\quad\xrightarrow{\;\eta\to0\;}\quad W\qquad \text{weakly in}\quad \fspaceL^2\at\Rset\,,
\end{align}
and this implies $\norm{W}_2\leq \liminf_{\eta\to0}\norm{\calA_\eta W}_2$. On the other hand, the estimate \eqref{Lem:PropertiesOperatorA.Eqn1}$_2$ ensures that $\limsup_{\eta\to0}\norm{\calA_\eta W}_2\leq\norm{W}_2$.
We therefore have $\norm{W}_2=\lim_{\eta\to0}\norm{\calA_\eta W}_2$ and combining this with the weak convergence 
\eqref{Lem:LimitOperatorA.PEqn0} we arrive at \eqref{Lem:LimitOperatorA.Eqn0} since $\fspaceL^2\at\Rset$ is a Hilbert space.
\end{proof}
%
%
\subsection{Asymptotic properties of the auxiliary operator \texorpdfstring{$\calB_\eps$}{}}
%
As already outlined above, we introduce for any given $\eps>0$ the operator
\begin{align}
\label{Eqn:OperatorBeps}
\calB_\eps:=  \id + \sum_{m=1}^M\al_m m^2 \frac{\id-\calA_{m\eps}^2}{\eps^2}\,,
\end{align}
which appears in \eqref{eq:scaledfpu1a} if we collect all linear terms on the left hand side, insert the wave-speed scaling \eqref{eq:speedscaling}, and divide the equation by $\eps^4$. We further define the operator
\begin{align}
\label{Eqn:OperatorB0}
\calB_0:=\id - \frac{\sum_{m=1}^M \al_m m^4}{12}\,\partial_x^2\,,
\end{align}
which can -- thanks to Lemma \ref{Lem:LimitOperatorA} -- be regarded as the formal limit of $\calB_\eps$ as $\eps\to0$. In Fourier space, these operators correspond to the symbol functions
\begin{align}
\label{Eqn:OperatorBSymb}
b_\eps\at{k} =1+  \sum_{m=1}^M\al_m m^2 \frac{1-\sinc^2\at{mk\eps/2}}{\eps^2}\,,\qquad
b_0\at{k} =1+  \frac{\sum_{m=1}^M \al_m m^4}{12}k^2\,,
\end{align}
which are illustrated in Figure~\ref{Fig2} and satisfy
\begin{align*}
b_\eps\at{k}\quad \xrightarrow{\eps\to0}\quad b_0\at{k}
\end{align*}
for any fixed $k\in\Rset$. This convergence, however, does not hold uniformly in $k$ since $\calB_\eps$ is a singular perturbation of $\calB_0$. Using the uniform positivity of these symbol functions, we easily demonstrate the existence of the inverse operators
\begin{align*}
\calB_\eps^{-1},\,\calB_0^{-1} \;:\;\fspaceL^{2}\at\Rset\to\fspaceL^{2}\at\Rset\,,
\end{align*}
where $\calB_0^{-1}$ maps actually into the Sobolev space $\fspace{W}^{2,2}\at\Rset$ and is hence smoothing because $1/b_0\at{k}$ decays quadratically at infinity. The inverse of $\calB_\eps$, however, is less regularizing because  $b_\eps\at{k}$ remains bounded as $k\to\pm\infty$. In order to obtain asymptotic estimates for $\calB_\eps^{-1}$, we introduce the cut-off operator 
\begin{align*}
\Pi_\eps\;:\;\fspaceL^2\at\Rset\to\fspaceL^2\at\Rset
\end{align*}
by defining its symbol function $\pi_\eps$ as follows
\begin{align}
\label{eqn.cutoff.def}
\pi_\eps \at{k}:=\left\{\begin{array}{lcl}
1&&\text{for \;$\abs{k}\leq\D \frac{4}{\eps}$}\,,\\0&&\text{else}\,.
\end{array}\right.
\end{align} 
One of our key technical results is the following characterization of $\calB_\eps^{-1}$, which reveals that $\calB_\eps$ admits an almost compact inverse. For $m=1$, a similar but slightly stronger result has been given in \cite[Corollary 3.5]{FP99} using a careful Fourier-pole analysis of the involved integral operators. For $m>1$, however, the symbol functions possess more poles in the complex plane and hence we argue differently.
\begin{figure}[t!]%
\centering{%
\includegraphics[width=0.8\textwidth]{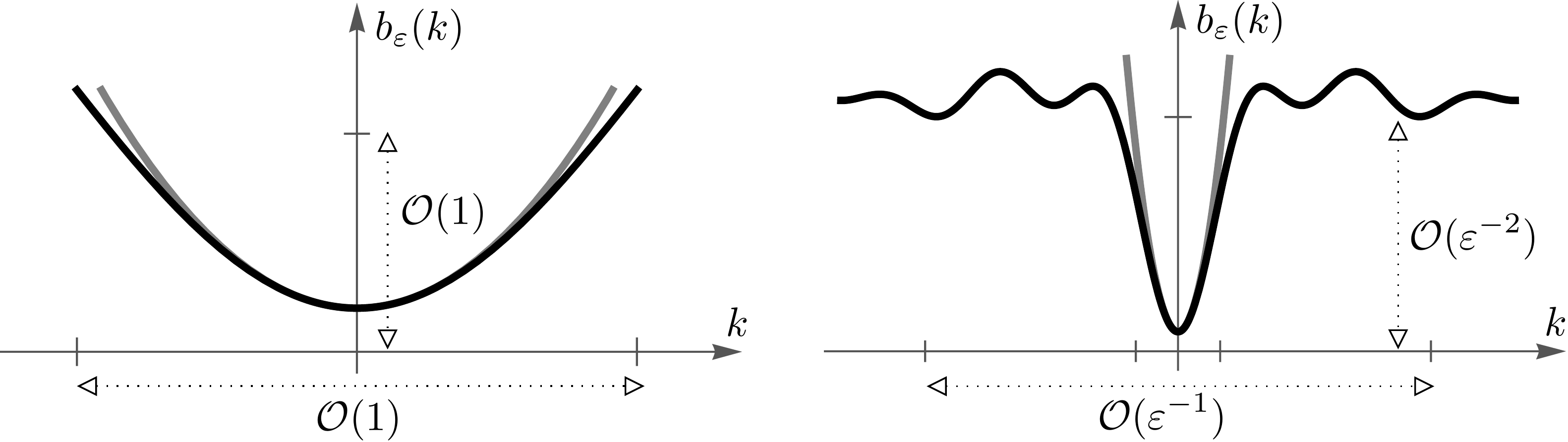}%
}%
\caption{Sketch of the symbol function $b_\eps$ from \eqref{Eqn:OperatorBSymb}, depicted on two intervals for $\eps>0$ (black) and $\eps=0$ (gray). Notice that $b_\eps\at{0}=\min_{k\in\Rset} b_\eps\at{k}=1$ holds for all $\eps\geq0$.
}%
\label{Fig2}%
\end{figure}%
\begin{lemma}[asymptotic estimates for $\calB_\eps^{-1}$]
\label{Lem:InversOfB}
For any $\eps>0$, the operator $\calB_\eps$ respects the even-odd parity and is both
self-adjoint and invertible on $\fspaceL^2\at\Rset$. Moreover, there exists a constant 
$C$ such that 
\begin{align}
\label{Lem:InversOfB.Eqn1} 
\norm{\Pi_\eps \calB_\eps^{-1} G}_{2,2}+\eps^{-2}\norm{\at{\id-\Pi_\eps}\calB_\eps^{-1} G}_{2}\leq C \norm{G}_{2}
\end{align} 
holds for all $G\in\fspaceL^2\at\Rset$ and all $0<\eps\leq1$. Here,
$\norm{\cdot}_{2,2}$ denotes the usual norm in  $\fspaceW^{2,2}\at\Rset$.
\end{lemma}
\begin{proof}
In view of \eqref{Eqn:OperatorBeps}, \eqref{Eqn:OperatorBSymb} and Lemma \ref{Lem:PropertiesOperatorA}, it remains to show \eqref{Lem:InversOfB.Eqn1}. Using the properties of the $sinc$ function, see Figure \ref{Fig1}, we readily demonstrate
\begin{align*}
1\geq 1-\sinc^2\at{mz}\geq \frac{\at{\min\{\abs{z},\,2\}}^2}{6}\qquad \text{for all} \quad z\in\Rset\quad\text{and} \quad m\in\Nset\,.
\end{align*}
Consequently, we get
\begin{align*}
1-\sinc^2\at{m\eps k /2}\geq 
\frac{1}{24}\left\{\begin{array}{lcl}%
\D\eps^2k^2&& \text{for $\;\abs{k}\leq \frac{4}{\eps}$}
\\
16&& \text{else}
\end{array}\right.%
\end{align*}
for all $m$, and hence
\begin{align*}
b_\eps\at{k}\geq  c \left\{\begin{array}{lcl}%
1+k^2&& \text{for $\;\abs{k}\leq \frac{4}{\eps}$}
\\
\D1/\eps^2&& \text{else}
\end{array}\right.%
\end{align*}
for some positive constant $c>0$. Moreover, noting that
\begin{align*}
\widehat{\calB_\eps^{-1}G}\at{k}= \frac{\widehat{G}\at{k}}{b_\eps\at{k}}
\end{align*}
and using Parseval's theorem we estimate
\begin{align*}
\norm{\Pi_\eps \calB_\eps^{-1} G}_{2,2}^2&= \int_{\abs{k}\leq \frac 4\eps}
\at{1+k^2+k^4}\frac{\babs{\widehat{G}(k)}^2}{b_\eps\at{k}^2}\dint k
\leq \frac{1}{c^2} 
\int_{\abs{k}\leq \frac4\eps}\frac{1+k^2+k^4}{1+2k^2+k^4}
\babs{\widehat{G}(k)}^2\dint{k}\leq
\frac{1}{c^2}\norm{G}_2^2
\end{align*}
as well as
\begin{align*}
\norm{\at{\id-\Pi_\eps}\calB_\eps^{-1} G}_{2}^2=
 \int_{\abs{k}\geq \frac{4}{\eps}}{\frac{\babs{\widehat{G}(k)}^2}{b_\eps\at{k}^2}}\dint k\leq \frac{\eps^4}{c^2} \norm{G}_2^2\,.
\end{align*}
so \eqref{Lem:InversOfB.Eqn1} follows immediately. 
\end{proof}
There exists another useful characterization of $\calB_\eps^{-1}$, which relies on the non-expansive estimate $\norm{\calA_{m\eps}W}_\infty \leq \norm{W}_\infty$, see Lemma \ref{Lem:PropertiesOperatorA}.
\begin{lemma}[von Neumann representation]
\label{Lem:vonNeumann}
We have
\begin{align*}
\calB_\eps^{-1}=\eps^2\sum_{i=0}^\infty \frac{\at{\sum_{m=1}^M\al_m m^2\calA_{m\eps}^2}^{i}}{\at{\eps^2+\sum_{m=1}^M \al_m m^2}^{i+1}}\,,
\end{align*}
where the series on the right hand converges for any $W\in\fspaceL^2\at\Rset$.
\end{lemma}
\begin{proof}
In the first step we regard all operators as defined on and taking values in $\fspaceL^\infty\at\Rset$. We also use the abbreviation
\begin{align*}
\calI_\eps := \frac{\sum_{m=1}^M \al_m m^2 \calA_{m\eps}^2}{\eps^2+c_0^2}
\end{align*}
and notice that \eqref{eq:speedscaling} and \eqref{Eqn:OperatorBeps} imply
\begin{align*}
\calB_\eps=
\frac{\eps^2+c_0^2}{\eps^2}\at{\mathrm{Id}-\calI_\eps}\,.
\end{align*}
Since the operator norm of $\calI_\eps$ -- computed with respect to the $\infty$-norm -- satisfies
\begin{align*}
\norm{\calI_\eps}_{{\mathrm{op}}}\leq\frac{c_0^2}{\eps^2+c_0^2}<1\,,
\end{align*}
the von Neumann formula provides
\begin{align}
\label{Lem:vonNeumann.PEqn1}
\calB_\eps^{-1}=\frac{\eps^2}{\eps^2+c_0^2}\Bat{\id +\calI_\eps+\calI_\eps^2+\tdots}=
\frac{\eps^2}{\eps^2+c_0^2}\id +\frac{\eps^2}{\eps^2+c_0^2}\Bat{\id+\calI_\eps+\calI_\eps^2+\tdots}\calI_\eps
\end{align}
in the sense of an absolutely convergent series of $\fspaceL^\infty$-operators. In the second step we generalize this result using the estimates from Lemma \ref{Lem:PropertiesOperatorA}.  In particular, the right-hand side in \eqref{Lem:vonNeumann.PEqn1} is well-defined for any $W\in\fspaceL^2\at\Rset$ since Lemma \ref{Lem:PropertiesOperatorA} ensures $\calI _\eps W\in\fspaceL^\infty\at\Rset$.
\end{proof}
\begin{corollary}[invariance properties of $\calB_\eps^{-1}$]
\label{Cor:InvarianceProperties}
The operator $\calB_\eps^{-1}$ respects for both $\eps>0$ and $\eps=0$ the 
nonnegativity, the evenness, and the unimodality of functions.
\end{corollary}
\begin{proof}
For $\eps>0$, all assertions follow from the representation formula in Lemma \ref{Lem:vonNeumann} and the corresponding properties of the operators $\calA_{m\eps}$, see Lemma \ref{Lem:PropertiesOperatorA}. For $\eps=0$ we additionally employ the approximation results from Lemma \ref{Lem:LimitOperatorA} as well as the estimates from Lemma \ref{Lem:InversOfB}.
\end{proof}
Note that all results concerning $\calB_\eps^{-1}$ are intimately related to the supersonicity condition $c_\eps^2> c_0^2$. In a subsonic setting, one can still establish partial inversion formulas but the analysis is completely different, cf. \cite{HMSZ13} for an application in a different context.
%
\section{Proof of the main result}\label{sect:proof}
%
%
In view of the wave-speed scaling \eqref{eq:speedscaling} and the fixed point formulation \eqref{eq:scaledfpu1a}, the rescaled traveling wave problem consists in finding solutions $W_\eps\in\fspaceL^2\at\Rset$ to the operator equation
\begin{align}
\label{Eqn:RescaledTWEqn}
\calB_\eps W_\eps  = \calQ_\eps\ato{W_\eps}+\eps^2 \calP_\eps\ato{W_\eps}\,,
\end{align}
where the linear operator $\calB_\eps$ has been introduced 
in \eqref{Eqn:OperatorBeps}. Moreover, the 
nonlinear operators 
\begin{align}
\label{Eqn:NonlinOp.W}
\calQ_\eps\ato{W}:= \sum_{m=1}^M \be_mm^3 \calA_{m\eps}\at{\calA_{m\eps}W}^2\,,\qquad
\calP_\eps\ato{W}:=
\frac{1}{\eps^6}\sum_{m=1}^M m \calA_{m\eps} \Psi_m^\prime\at{m \eps^2 \calA_{m\eps} W}
\end{align}
encode the quadratic and cubic nonlinearities, respectively, and are scaled such that the respective formal $\eps$-expansions involve nontrivial leading order terms. In particular, we have
\begin{align}
\label{Eqn:NonlinOp.Q}
\calQ_\eps\ato{W}\quad\xrightarrow{\;\;\eps\to0\;\;}\quad
\calQ_0\ato{W}:=\at{\sum_{m=1}^M \be_mm^3} W^2\,,
\end{align}
for any fixed $W\in\fspaceL^2\at\Rset$, see \eqref{Lem:LimitOperatorA.Eqn0}. Note also 
that \eqref{Eqn:RescaledTWEqn} always admits the trivial solution $W_\eps\equiv0$.
\par
In what follows we solve the leading order problem to obtain the KdV wave $W_0$, transform \eqref{Eqn:RescaledTWEqn} via the ansatz \eqref{eqn.def.corrector} into another fixed point equation, and employ the contraction mapping principle to prove the existence of a corrector $V_\eps\in\fspaceL^2_\even\at\Rset$ for all sufficiently small $\eps>0$. In \cite{FP99}, the last step has been solved using a operator-valued variant of the implicit function theorem.
%
%
\subsection{The leading order problem and the KdV wave}\label{sect:proof.1}
%
%
Passing formally to limit $\eps\to0$ in \eqref{Eqn:RescaledTWEqn}, we obtain the leading order equation
\begin{align}
\label{Eqn:LeadingOrderProblem.1}
\calB_0 W_0 = \calQ_0\ato{W_0}\,,
\end{align}
which is the ODE \eqref{Eqn:WaveODE} with parameters
\begin{align}
\label{Eqn:LeadingOrderProblem.2x}
d_1 := \frac{12}{\sum_{m=1}^M \al_m m^4}\,,\qquad
d_2 := \frac{12\sum_{m=1}^M \be_mm^3}{\sum_{m=1}^M \al_m m^4}\,.
\end{align}
In particular, the leading order problem is a planar Hamiltonian ODE with conserved quantity $E=\tfrac12\at{W^\prime}^2+\tfrac13d_2W^3-\tfrac12d_1W^2$
and admits precisely one homoclinic orbit as shown in  Figure \ref{Fig3}. 
\begin{figure}[t!]%
\centering{%
\includegraphics[width=0.8\textwidth]{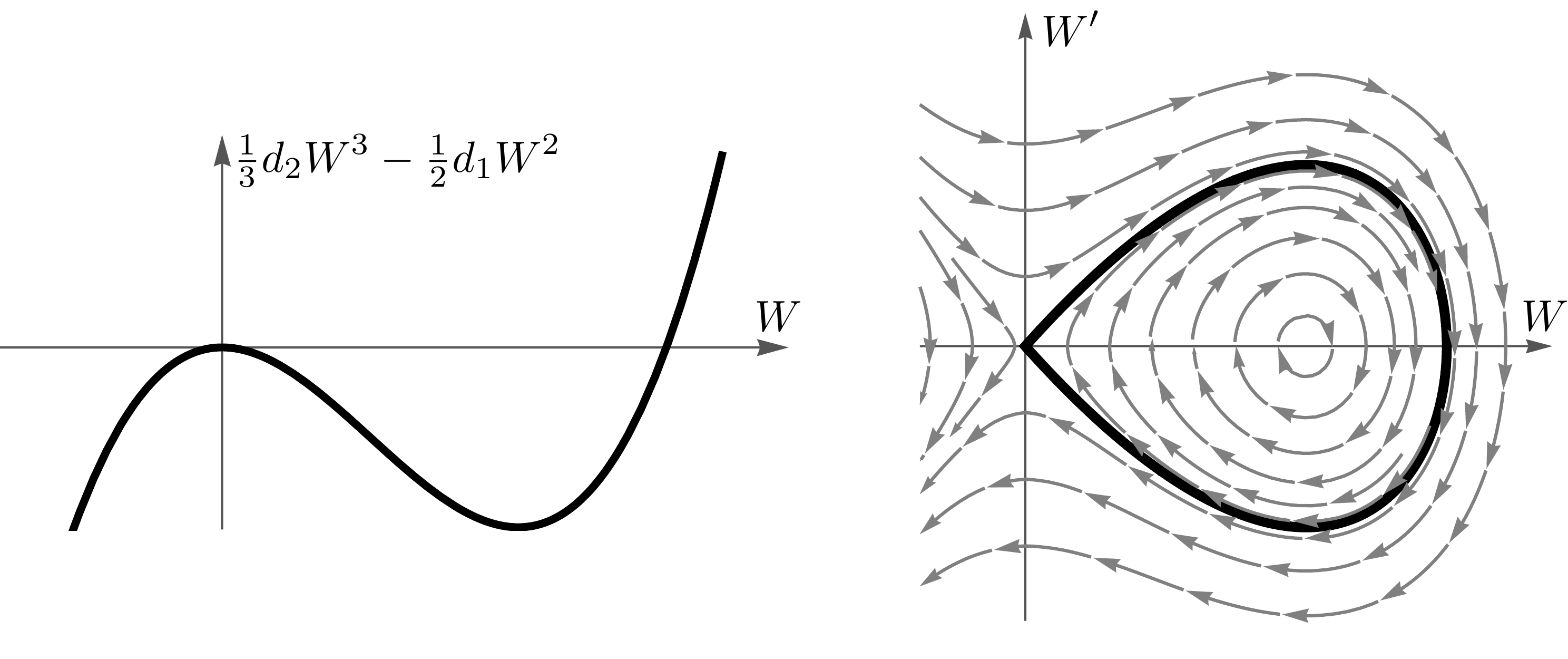}%
}%
\caption{%
Potential energy (\emph{left panel}) and phase diagram (\emph{right panel}) for the nonlinear oscillator ODE \eqref{Eqn:WaveODE} with coefficients \eqref{Eqn:LeadingOrderProblem.2x}, which determines the KdV wave $W_0$. There exists precisely one homoclinic orbit (solid black curve in the right panel) which corresponds to the solitary wave $W_0$. The closed loops inside the homoclinic orbits correspond to periodic KdV waves, see \cite{FM14}.
}%
\label{Fig3}%
\end{figure}%
\begin{lemma}[linear and nonlinear leading-order problem]
\label{Lem:LeadingOrder}
There exists a unique solution $W_0\in\fspaceL^2_\even\at\Rset$ to \eqref{Eqn:LeadingOrderProblem.1}, which is 
moreover smooth, pointwise positive, and exponentially decaying.  Moreover, the $\fspaceL^2$-kernel of the linear operator $\calL_0$ with
\begin{align}
\label{Lem:LeadingOrder.Eqn1}
\calL_0 V:= \calB_0 V- \calM_0V\,,\qquad \calM_0V:=2\at{\sum_{m=1}^M \beta_m m^3}W_0 V
\end{align}
is simple and spanned by the odd function $W_0^\prime$.
\end{lemma}
\begin{proof}
The existence and uniqueness of $W_0$ follow from standard ODE arguments and the identity
$\calL_0W_0^\prime=0$ holds by construction. Moreover, the simplicity of the  $\fspaceL^2$-kernel of the differential operator $\calL_0$ can be proven by the following Wronski-type argument: Suppose for contradiction that $V_1, V_2\in\fspaceL^2\at\Rset$ are two linearly independent kernel functions of $\calL_0$ such that $\om\at{0}\neq0$, where
\begin{align*}
,\qquad \om\at{x}:=\det\begin{pmatrix} V_1\at{x}& V_2\at{x}\\
V_1^\prime\at{x}& V_2^\prime\at{x}
\end{pmatrix}\,.
\end{align*} 
The ODE $\calL_0V_i=0$ combined with $V_i\in\fspaceL^2\at\Rset$ implies that $V_i$ and $V_i^\prime$ are continuous functions with
\begin{align*}
\abs{V_i\at{x}}+\abs{V_i^\prime\at{x}}\quad\xrightarrow{\;\abs{x}\to\infty\;}\quad0\,,
\end{align*}
and we conclude that $\om\at{x}\to0$ as $\abs{x}\to\infty$. On the other hand, we easily compute $\om^\prime\at{x}=0$ and obtain the desired contradiction.
\end{proof}
Since $W_0$ is smooth, it satisfies \eqref{Eqn:RescaledTWEqn} up to small error terms.
In particular, the corresponding linear and the quadratic terms almost cancel due to \eqref{Eqn:LeadingOrderProblem.1}.
\begin{lemma}[$\eps$-residual of $W_0$]
\label{Lem.epsResidual} 
There exists a constant $C$ such that 
\begin{align}
\label{Lem.epsResidual.Eqn1} 
\norm{R_\eps}_2+\norm{S_\eps}_2\leq C\qquad \text{with}\qquad R_\eps : =\frac{\calQ_\eps\ato{W_0}-\calB_\eps W_0}{\eps^2} \,,\qquad S_\eps : =\calP_\eps\ato{W_0}
\end{align}
holds for all $0<\eps\leq 1$.
\end{lemma}
\begin{proof}
We first notice that Lemma \ref{Lem:PropertiesOperatorA} ensures
\begin{align*}
\norm{\calA _{m\eps}^2W_0}_2\leq\norm{\calA _{m\eps}W_0}_2\leq\norm{W_0}_2\,,\qquad
\norm{\calA _{m\eps}W_0}_\infty\leq \norm{W_0}_\infty
\end{align*}
and in view of Assumption \ref{MainAssumption} we find
\begin{align*}
\norm{S_\eps}_2&\leq\frac{1}{\eps^6}\sum_{m=1}^M m \ga_m\norm{m\eps^2\calA_{m\eps}W_0}_\infty^2\norm{m\eps^2\calA_{m\eps}^2W_0}_2
\leq  C\,.
\end{align*}
Thanks to the smoothness of $W_0$, Lemma \ref{Lem:LimitOperatorA} 
provides a constant $C$ such that
\begin{align*}
\norm{\calA_{m\eps}W_0^j-W_0^j}_2+\norm{\calA_{m\eps}W_0^j-W_0^j}_\infty \leq Cm^2\eps^2
\end{align*}
holds for $j\in\{1,2\}$, and this implies
\begin{align*}
\Bnorm{\calA_{m\eps}\bat{\calA_{m\eps}W_0}^2-W_0^2}_2&\leq
\Bnorm{\calA_{m\eps}\bat{\calA_{m\eps}W_0}^2-\calA_{m\eps}W_0^2}_2+\Bnorm{\calA_{m\eps}W_0^2-W_0^2}_2
\\&\leq
\Bnorm{\bat{\calA_{m\eps}W_0}^2-W_0^2}_2+
Cm^2\eps^2
\\&\leq\bat{\norm{\calA_{m\eps}W_0}_\infty+\norm{W_0}_\infty}\norm{\calA_{m\eps}W_0-W_0}_2+
Cm^2\eps^2
\\&\leq Cm^2\eps^2
\end{align*}
and hence
\begin{align*}
\bnorm{\calQ_\eps\ato{W_0}-\calQ_0\ato{W_0}}_2=
\norm{\sum_{m=1}^M \be_m m^3 \calA_{m\eps}\at{\calA_{m\eps}W_0}^2-\at{\sum_{m=1}^{M}\be_m m^3} W_0^2}_2\leq C\eps^2.
\end{align*}
Therefore, and since $W_0$ satisfies \eqref{Eqn:LeadingOrderProblem.1}, we get
\begin{align}
\label{Lem.epsResidual.PEqn1}
\norm{R_\eps}_2&\leq \frac{\norm{\calB_\eps W_0- \calB_0W_0}_2}{\eps^2}+C
\leq 
\sum_{m=1}^M\al_mm^2\frac{\Bnorm{\calA_{m\eps}^2 W_0 - W_0- \frac{m^2\eps^2}{12}W_0^{\prime\prime}}_2}{\eps^4}+C\,,
\end{align}
where the second inequality stems from the definitions of $\calB_\eps$ and $\calB_0$, see \eqref{Eqn:OperatorBeps} and \eqref{Eqn:OperatorB0}.
Lemma \ref{Lem:LimitOperatorA} also yields
\begin{align*}
\norm{\calA_{m\eps}W_0-W_0-\frac{\eps^2m^2}{24}W_0^{\prime\prime}}_2\leq Cm^4\eps^4\,,\qquad
\norm{\calA_{m\eps}W_0^{\prime\prime}-W_0^{\prime\prime}}_2\leq Cm^2\eps^2
\end{align*}
and combining this with \eqref{Lem:PropertiesOperatorA.Eqn1}$_2$ and the identity
\begin{align*}
\calA_{m\eps}^2 W_0 - W_0- \frac{m^2\eps^2}{12}W_0^{\prime\prime}
&=\at{\calA_{m\eps}+\mathrm{id}}\Bat{\calA_{m\eps} W_0 - W_0- \frac{m^2\eps^2}{24}W_0^{\prime\prime}}+\frac{m^2\eps^2}{24}\Bat{\calA_{m\eps}W_0^{\prime\prime}-W_0^{\prime\prime}}\,,
\end{align*}
we arrive at
\begin{align*}
\norm{\calA_{m\eps}^2W_0-W_0-\frac{\eps^2m^2}{24}W_0^{\prime\prime}}_2&\leq 2\norm{\calA_{m\eps} W_0 - W_0- \frac{m^2\eps^2}{24}W_0^{\prime\prime}}_2+\frac{m^2\eps^2}{24}\norm{\calA_{m\eps}W_0^{\prime\prime}-W_0^{\prime\prime} }_2
\\&\leq Cm^4\eps^4\,.
\end{align*}
The desired estimate for $R_\eps$ is now a direct consequence of
\eqref{Lem.epsResidual.PEqn1}.
\end{proof}
For completeness we mention that 
\begin{align*}
W_0\at{x}=\frac{3d_1}{2d_2}\,{\mathrm{sech}}^2\at{\tfrac12\sqrt{d_1}{x}}
\end{align*}
can be verified by direct calculations and that formulas for the spectrum of $\calL_0$ can, for instance, be found in \cite[page 768]{MF53}; see also \cite[Lemma 4.2]{FP99}.
%
\subsection{The linearized traveling wave equation \texorpdfstring{for $\eps>0$}{}}\label{sect:proof.2}
%
%
For any $\eps>0$, we define the linear operator $\calL_\eps$ on $\fspaceL^2\at\Rset$ by
\begin{align}
\label{Eqn:DefLandM}
\calL_\eps V:= \calB_\eps V-\calM_\eps V\,,\qquad \calM_\eps V:= 2\sum_{m=1}^M \beta_m m^3 \calA_{m\eps}\Bat{\at{\calA_{m\eps}W_0}\at{\calA_{m\eps} V}},
\end{align}
where $W_0\in\fspaceL^2_\even\at\Rset$ is the unique even KdV wave provided by Lemma \ref{Lem:LeadingOrder}. This operator appears naturally in the linearization \eqref{Eqn:RescaledTWEqn} around $W_0$ as
\begin{align*}
\calB_\eps\at{W_0+\eps^2V}-\calQ_\eps\ato{W_0+\eps^2V} = - \eps^2 R_\eps + \eps^2 \calL_\eps V - \eps^4 \calQ_\eps\at{V}
\end{align*}
holds due to the linearity of $\calB_\eps$ and the quadraticity of $\calQ_\eps$. 
\begin{lemma}[elementary properties of $\calL_\eps$]
\label{Lem:PropertiesOfL}
For any $\eps>0$, the operator $\calL_\eps$ is self-adjoint in $\fspaceL^2\at\Rset$ and respects the even-odd parity. Moreover, we have
\begin{align*}
\calL_\eps W\quad\xrightarrow{\;\eps\to0\;}\quad \calL_0 W\qquad \text{strongly in}\quad \fspaceL^2\at\Rset
\end{align*}
for any $W\in\fspaceW^{2,2}\at\Rset$. 
\end{lemma}
\begin{proof}
Since $W_0$ is smooth and even, all assertions follow directly from the properties
of $\calA_{m\eps}$ and $\calB_\eps$, see \eqref{Eqn:OperatorBeps} and Lemma \ref{Lem:LimitOperatorA}.
\end{proof}
Our perturbative approach requires to invert the operator $\calL_\eps$ on the space $\fspaceL^2_\even\at\Rset$ -- see the fixed point problem in Theorem \ref{Thm:FixedPoints} below -- and in view of Lemma \ref{Lem:LeadingOrder} one easily shows that $\calL_0$ has this properties. The singularly perturbed case $\eps>0$, however, is more involved and addressed in the following theorem,  which is 
actually the key asymptotic result in our approach.  Notice that the analogue for $M=1$ is not stated explicitly in \cite{FP99} although it could be derived from the asymptotic estimates therein.
\begin{theorem}[uniform invertibility of $\calL_\eps$]
\label{Lem:InvertibilityOfLeps}
There exists $0<\eps_*\leq1$ such that for any $0<\eps\leq\eps_*$ the operator $\calL_\eps$ is continuously invertible on $\fspaceL^2_\even\at\Rset$. More precisely, there exists a constant $C$ which depends on $\eps_*$ but not on $\eps$ such that
\begin{align*}
\norm{\calL_\eps^{-1}G}_2\leq C\norm{G}_2
\end{align*}
holds for all $0<\eps\leq\eps_*$ and any $G\in\fspaceL^2_\even\at\Rset$.
\end{theorem}
\begin{proof} 
\ul{\emph{Preliminaries}}: Our strategy is to
show the existence of a constant $c_*>0$ such that
\begin{align}
\label{Lem:InvertibilityOfLeps.PEqn10}
\norm{\calL_\eps V}_2\geq c_* \norm{V}_2
\end{align}
holds for all $V\in\fspaceL^2_\even\at\Rset$ and all sufficiently small $\eps>0$, because this implies the desired result. In fact, \eqref{Lem:InvertibilityOfLeps.PEqn10} ensures that the operator
\begin{align*}
\calL_\eps:\fspaceL^2_\even\at\Rset\to \fspaceL^2_\even\at\Rset
\end{align*} 
has both trivial kernel and closed image. The symmetry of $\calL_\eps$ gives
\begin{align*}
\ker \calL_\eps =  \mathrm{coker}\, \calL_\eps
\end{align*}
and due to the closed image we conclude that $\calL_\eps$ is not only injective but also surjective. Moreover,
the $\eps$-uniform continuity of the inverse $\calL_\eps^{-1}$ is a further consequence of \eqref{Lem:InvertibilityOfLeps.PEqn10}.
\par
Now suppose for contradiction that such a constant $c_*$ does not exist. Then we can choose a sequence $\at{\eps_n}_{n\in\Nset}\subset\ocinterval{0}{1}$ with $\eps_n\to0$ as well as sequences $\at{V_n}_{n\in\Nset}\subset\fspaceL^2_\even\at\Rset$ and $\at{G_n}_{n\in\Nset}\subset\fspaceL^2_\even\at\Rset$ such that
\begin{align}
\label{Lem:InvertibilityOfLeps.PEqn8}
\calL_{\eps_n}V_n=G_n\,,\qquad \norm{V_n}_2=1\,,\qquad \norm{G_n}_2\quad\xrightarrow{n\to\infty}\quad0\,.
\end{align}
\par
\ul{\emph{Weak convergence to $0$}}:
By weak compactness we can assume that there exists $V_\infty\in\fspaceL^2_\even\at\Rset$ such that
\begin{align}
\label{Lem:InvertibilityOfLeps.PEqn1}
V_n\quad \xrightharpoonup{\;n\to\infty\;}\quad V_\infty\qquad \text{weakly in $\fspaceL^2\at\Rset$}\,,
\end{align} 
and using Lemma \ref{Lem:PropertiesOfL}  we find
\begin{align*}
\skp{V_\infty}{\calL_{0}\phi}=
\lim_{n\to\infty}\skp{V_n}{\calL_{\eps_n}\phi}=
\lim_{n\to\infty}\skp{\calL_{\eps_n}V_n}{\phi}=
\lim_{n\to\infty}\skp{G_n}{\phi}=0
\end{align*}
for any sufficiently smooth test function $\phi$.  In view of
the definition of the differential operator $\calL_0$ -- see \eqref{Eqn:OperatorB0}  and \eqref{Lem:LeadingOrder.Eqn1} -- we estimate
\begin{align*}
\abs{\int_\Rset W_0\at{x}\phi^{\prime\prime}\at{x}\dint{x}}\leq C \norm{\phi}_2
\end{align*}
for all $\phi\in\fspaceW^{2,2}\at\Rset$ and conclude that $V_\infty$ belongs to $\fspaceW^{2,2}\at\Rset$, where $\calL_0V_\infty=0$ holds due to 
\begin{align*}
\skp{\calL_{0} V_\infty}{\phi}=\skp{V_\infty}{\calL_{0}\phi}=0\,.
\end{align*}
In other words, the even function $V_\infty$ belongs to the kernel of $\calL_0$ and
\begin{align*}
V_\infty=0
\end{align*}
follows from Lemma \ref{Lem:LeadingOrder}.
\par
\ul{\emph{Further notations:}}
For the remaining considerations we abbreviate the constant from Lemma \ref{Lem:InversOfB} by $D$ and denote by $C$ any generic constant (whose value may change from line to line) that is independent of $n$ and $D$. We further choose $K>M$ sufficiently large such that
\begin{align}
\label{Lem:InvertibilityOfLeps.PEqn5a}
\sup\limits_{\abs{\xi}\geq K-M} W_0\at{\xi}\leq \frac{1}{4 D\sum_{m=1}^M \be_m m^3}\,,
\end{align}
and denote by $\chi_K$ the characteristic function of the interval
$I_K:=\ccinterval{-K}{{+}K}$. In what follows we write $V_n=V_n^\upidx{1}+V_n^\upidx{2}+V_n^\upidx{3}$ with
\begin{align*}
V_n^\upidx{1} := \chi_K \,\Pi_{\eps_n} V_n\,,\qquad
V_n^\upidx{2} := \at{1-\chi_K} \,\Pi_{\eps_n} V_n\,,\qquad
V_n^\upidx{3} :=\at{\mathrm{id}- \Pi_{\eps_n}}V_n
\end{align*}
and observe that these definitions imply
\begin{align}
\label{Lem:InvertibilityOfLeps.PEqn5b}
\max\limits_{i\in\{1,2,3\}}\bnorm{V_n^\upidx{i}}_{2}\leq \norm{V_n}_2=1\,.
\end{align}
We also set
\begin{align*}
U_n^\upidx{i}:=\calM_{\eps_n} V_n^\upidx{i}
\end{align*} 
and combine Lemma \ref{Lem:PropertiesOperatorA} with the smoothness of $W_0$ to obtain
\begin{align}
\label{Lem:InvertibilityOfLeps.PEqn7}
\bnorm{U_n^\upidx{i}}_2\leq C\bnorm{V_n^\upidx{i}}_2\,.
\end{align}
Moreover, by construction we have
\begin{align*}
V_n=\calB_{\eps_n}^{-1}\Bat{U_n^\upidx{1}+U_n^\upidx{2}+U_n^\upidx{3}+G_n}\,,
\end{align*}
so the estimate
\begin{align}
\label{Lem:InvertibilityOfLeps.PEqn6a}
\bnorm{V_n^\upidx{1}+V_n^\upidx{2}}_{2,2} + \eps_n^{-2} \bnorm{V_n^\upidx{3}}_{2}\leq D \at{\bnorm{U_n^\upidx{1}}_2+\bnorm{U_n^\upidx{2}}_2+\bnorm{U_n^\upidx{3}}_2 + \norm{G_n}_2}
\end{align}
is provided by Lemma \ref{Lem:InversOfB}.
\par
\ul{\emph{Strong convergence of $V_n^\upidx{1}$ and  $V_n^\upidx{3}$}}: 
Inserting \eqref{Lem:InvertibilityOfLeps.PEqn8}, \eqref{Lem:InvertibilityOfLeps.PEqn5b}, and \eqref{Lem:InvertibilityOfLeps.PEqn7} into \eqref{Lem:InvertibilityOfLeps.PEqn6a} gives
\begin{align}
\label{Lem:InvertibilityOfLeps.PEqn6b}
\bnorm{V_n^\upidx{1}+V_n^\upidx{2}}_{2,2} + \eps_n^{-2} \bnorm{V_n^\upidx{3}}_{2}\leq CD
\end{align}
and hence
\begin{align}
\label{Lem:InvertibilityOfLeps.PEqn2}
V_n^\upidx{3}\quad\xrightarrow{n\to\infty}\quad 0\qquad\text{strongly in $\fspaceL^2\at\Rset$}\,.
\end{align}
Thanks to
\begin{align}
\label{Lem:InvertibilityOfLeps.PEqn2a}
V_n^\upidx{2}=0\qquad \text{in}\qquad \fspaceL^2\at{I_K}
\end{align}
we also infer from \eqref{Lem:InvertibilityOfLeps.PEqn6b} the estimate
\begin{align*}
\bnorm{V_n^\upidx{1}}_{2,2, I_K}\leq
\bnorm{V_n^\upidx{1}+V_n^\upidx{2}}_{2,2}\leq CD\,,
\end{align*}
where $\bnorm{\cdot}_{2,2, I_K}$ denotes the norm in $\fspaceW^{2,2}\at{I_K}$.
Since $\fspaceW^{2,2}\at{I_K}$ is compactly embedded into $\fspaceL^{2}\at{I_K}$, we conclude that the sequence $\bat{V_n^\upidx{1}}_{n\in\Nset}$ is precompact in $\fspaceL^2\at{I_K}$. On other hand, the weak convergence \eqref{Lem:InvertibilityOfLeps.PEqn1} combined with \eqref{Lem:InvertibilityOfLeps.PEqn2} and \eqref{Lem:InvertibilityOfLeps.PEqn2a}
implies
\begin{align*}
V_n^\upidx{1}\quad \xrightharpoonup{\;n\to\infty\;}\quad V_\infty=0\qquad \text{weakly in}\quad \fspaceL^2\at {I_K},
\end{align*} 
and in summary we find $V_n^\upidx{1}\to0$ strongly in $\fspaceL^2\at{I_K}$ by standard arguments.
This even implies
\begin{align}
\label{Lem:InvertibilityOfLeps.PEqn3}
V_n^\upidx{1}\quad\xrightarrow{n\to\infty}\quad 0\qquad\text{strongly in $\fspaceL^2\at{\Rset}$}
\end{align}
as $V_n^\upidx{1}$  vanishes outside the interval $I_K$.
\par
\ul{\emph{Upper bounds for $\nnorm{U_n^\upidx{2}}_2$}}: %
Since the functions $V_n^\upidx{2}$ are supported in $\Rset\setminus I_K$, the functions $\calA_{m\eps_n}V_n^\upidx{2}$ are supported in $\Rset\setminus I_{K-m\eps_n/2}=\{{x}:\abs{{x}}>K-m\eps_n/2\}$. Moreover, we have
\begin{align*}
\babs{\at{\calA_{m\eps_n}W_0}\at{\xi}}\;\leq
\sup_{\abs{{x}-\xi}\leq m\eps_n/2} {W_0\at{x}}\;\leq
\sup_{\abs{{x}-\xi}\leq M/2} {W_0\at{x}}
\end{align*}
for any given $\xi\in\Rset$. Therefore, and using 
\begin{align*}
\Babs{\bat{\calA_{m\eps} V_n^\upidx2}\at{x}}\leq \Bat{\calA_{m\eps} \babs{V_n^\upidx2}}\at{x}\qquad\text{for all}\quad x\in\Rset\,,
\end{align*}
we estimate
\begin{align*}
\abs{\bat{\calA_{m\eps_n}W_0}\bat{\calA_{m\eps_n}V_n^\upidx2}}&\leq \at{\sup\limits_{\abs{\xi}\geq K-M/2}\babs{\at{\calA_{m\eps_n}W_0}\at\xi}}
\babs{\calA_{m\eps_n}V_n^\upidx{2}}\\&\leq
\at{\sup\limits_{\abs{\xi}\geq K-M}{W_0\at\xi}}
\calA_{m\eps_n}\babs{V_n^\upidx{2}}\,,
\end{align*}
so Lemma \ref{Lem:PropertiesOperatorA} gives
\begin{align*}
\norm{\calA_{m\eps_n}\at{\bat{\calA_{m\eps_n}W_0}\bat{\calA_{m\eps_n}V_n^\upidx2}}}_2\leq
\at{\sup\limits_{\abs{\xi}\geq K-M}{W_0\at\xi}}
\bnorm{V_n^\upidx{2}}_2
\end{align*}
and hence
\begin{align}
\label{Lem:InvertibilityOfLeps.PEqn4}
\bnorm{U_n^\upidx{2}}_2\leq \at{\sup\limits_{\abs{\xi}\geq K-M}{W_0\at\xi}}\at{2\sum_{m=1}^M\be_mm^3}\bnorm{V_n^\upidx{2}}_2\leq
\frac{1}{2D}
\end{align}
due to \eqref{Lem:InvertibilityOfLeps.PEqn5a} and \eqref{Lem:InvertibilityOfLeps.PEqn5b}.
\par
\ul{\emph{Derivation of the contradiction}}: 
Combining \eqref{Lem:InvertibilityOfLeps.PEqn6a} with
\eqref{Lem:InvertibilityOfLeps.PEqn7} gives
\begin{align*}
\bnorm{V_n}_2&\leq 
\bnorm{V_n^\upidx{1}+V_n^\upidx{2}}_2+\bnorm{V_n^\upidx{3}}_2
\\&\leq
D\at{\bnorm{U_n^\upidx{1}}_2+
\bnorm{U_n^\upidx{2}}_2+\bnorm{ U_n^\upidx{3}}_2+
\bnorm{ G_n}_2}
\\&\leq
D\Bat{C\bnorm{V_n^\upidx{1}}_2+\bnorm{U_n^\upidx{2}}_2+C\bnorm{V_n^\upidx{3}}_2+\bnorm{ G_n}_2} \,,
\end{align*}
and passing to the limit $n\to\infty$ we get
\begin{align*}
\limsup_{n\to\infty}\norm{V_n}_2\leq D\limsup_{n\to\infty} \bnorm{U_n^\upidx{2}}_2\leq\tfrac12
\end{align*}
thanks to \eqref{Lem:InvertibilityOfLeps.PEqn8}$_3$, \eqref{Lem:InvertibilityOfLeps.PEqn2},
\eqref{Lem:InvertibilityOfLeps.PEqn3}, and \eqref{Lem:InvertibilityOfLeps.PEqn4}. This, however, contradicts the normalization condition \eqref{Lem:InvertibilityOfLeps.PEqn8}$_2$. In particular, we have shown the existence of a constant $c_*$ as in \eqref{Lem:InvertibilityOfLeps.PEqn10} and the proof is complete.
\end{proof}
%
\subsection{Nonlinear fixed point argument}\label{sect:proof.3}
%
%
Setting $W_\eps=W_0+\eps^2V_\eps$, the nonlocal traveling wave equation \eqref{Eqn:RescaledTWEqn} is equivalent to
\begin{align*}
\calL_\eps V_\eps = R_\eps + S_\eps+\eps^2 \,\calQ_\eps\ato{V_\eps}+\eps^2\,\calN_\eps\ato{V_\eps}
\end{align*} 
with 
\begin{align}
\label{Eqn:NonlinOP.V}
\calN_\eps\ato{V}:=
\frac{ \calP_\eps \ato{W_0+\eps^2 V}-\calP_\eps\ato{W_0}}{\eps^2}\,,
\end{align}
where  $\calQ_\eps$, $\calP_\eps$  and  $R_\eps$,  $S_\eps$ have been introduced in \eqref{Eqn:NonlinOp.W} and \eqref{Lem.epsResidual.Eqn1}, respectively.
Since $\calL_\eps$ can be inverted for all sufficiently small $\eps>0$, we finally arrive at
the following result.
\begin{theorem}
[existence and uniqueness of the corrector $V_\eps$] 
\label{Thm:FixedPoints}
There exist constants $D>0$ and $0<\eps_*\leq1$ such that
the nonlinear operator $\calF_\eps$ with
\begin{align}
\label{Thm:FixedPoints.Eqn1}
\calF_\eps \ato{V} := \calL_\eps^{-1}\Bat{R_\eps + S_\eps + \eps^2\, \calQ_\eps\ato{V}+\eps^2\, \calN_\eps\ato{V}}
\end{align}
admits for any $0< \eps\leq\eps_*$ a unique fixed point $V_\eps$ in the set $B_D=\{V:\fspaceL^2_\even\at\Rset\;:\; \norm{V}_2\leq{D}\}$.
\end{theorem}
\begin{proof}
Our strategy is to demonstrate that the operator $\calF_\eps$ maps $B_D$ contractively into itself provided that $D$ is sufficiently large and $\eps$ sufficiently small; the desired result is then a direct consequence of the Banach fixed-point theorem. Within this proof we denote by $C$ any generic constant that is independent of $D$ and $\eps$. We also observe that $\abs{\calA_\eta Z}\leq \calA_\eta\abs{Z}$ holds for any $Z\in\fspaceL^1_{\mathrm{loc}}\at\Rset$ and $\eta>0$, and recall that
\begin{align*}
\norm{R_\eps+S_\eps}_2\leq C
\end{align*}
is provided by Lemma \ref{Lem.epsResidual}.
\par 
\emph{\ul{Estimates for the quadratic terms}}: 
For $V\in B_D$ we find
\begin{align*}
\babs{\eps^2 \calQ_\eps\ato{V}}\leq\eps^2 \sum_{m=1}^M
\be_m  m^3 \norm{\calA_{m\eps}V}_\infty \calA_{m\eps}^2\abs{V}\leq 
\eps^{3/2}\at{\sum_{m=1}^M \beta_mm^{5/2} D}\calA_{m\eps}^2\babs{V}\,,
\end{align*}
where we used the estimate \eqref{Lem:PropertiesOperatorA.Eqn1}$_1$, and
in view of \eqref{Lem:PropertiesOperatorA.Eqn1}$_2$  we obtain
\begin{align*}
\bnorm{\eps^2 \calQ_\eps\ato{V}}_2
\leq \eps^{3/2} C D \norm{\calA_{m\eps}^2V}_2
\leq \eps^{3/2} C D \norm{V}_2\leq\eps^{3/2} C D^2.
\end{align*}
In the same way we verify the estimate
\begin{align*}
\bnorm{\eps^2 \calQ_\eps\ato{V_2}-\eps^2 \calQ_\eps\ato{V_2}}_2&\leq \norm{\eps^2\sum_{m=1}^M
\be_m  m^3 \Bat{\norm{\calA_{m\eps}V_2}_\infty+
\norm{\calA_{m\eps}V_1}_\infty}\calA_{m\eps}^2\babs{V_2-V_1}}_2
\\&\leq 
\eps^{3/2}CD\norm{\calA_{m\eps}^2\abs{V_2-V_1}}_2\leq
\eps^{3/2}CD\norm{V_2-V_1}_2
\end{align*}
for arbitrary $V_1,V_2\in B_D$.
\par 
\emph{\ul{Estimates for the higher order terms}}:
For $V_1,V_2\in B_D$ we set
$Z_{m,\eps,i}:=\eps^2m \calA_{m\eps}\at{W_0+\eps^2 V_i}$ and employ  \eqref{Lem:PropertiesOperatorA.Eqn1}$_1$ to estimate
\begin{align*}
\norm{Z_{m,\eps,i}}_\infty&\leq
\eps^2m\norm{\calA_{m\eps}W_0}_\infty+
\eps^4m\norm{\calA_{m\eps}V_i}_\infty
\\&\leq
\eps^2m\norm{W_0}_\infty+
\eps^{7/2}m^{1/2}\norm{V_i}_2
\\&\leq
\eps^2 m\at{C+\eps^{3/2}D}=:\zeta_{m,\eps}\,.
\end{align*}
Due to the intermediate value theorem as well as the properties of $\Psi_m^{\prime\prime}$ we get
\begin{align*}
\Babs{\eps^2\calN_\eps\ato{V_2}-\eps^2\calN_\eps\ato{V_1}}&\leq
\sum_{m=1}^M m\abs{\frac{\Psi^\prime_m\bat{Z_{m,\eps,2}}-\Psi^\prime_m\bat{Z_{m,\eps,1}}}{\eps^6}}
\\&
\leq
\sum_{m=1}^M  
\frac{m\ga_m \zeta _{m,\eps}^2 \abs{Z_{m,\eps,2}-Z_{m,\eps,1}}}{\eps^6}
\\&
\leq
\sum_{m=1}^M 
\frac{m^2\ga_m \zeta _{m,\eps}^2 \abs{\calA_{m\eps}V_2 - \calA_{m\eps}V_1}}{\eps^2}
\\&\leq \eps^2\at{C+\eps^{3/2}D}^2\at{\sum_{m=1}^M
\ga_m m^{4}}
\calA_{m\eps}\babs{V_2-V_1}
\end{align*}
and hence
\begin{align*}
\bnorm{\eps^2 \calN_\eps\ato{V_2}-\eps^2\calN_\eps\ato{V_1}}_2\leq
\eps^2C \at{C+\eps^{3/2}D}^2\bnorm{V_2-V_1}_2\,
\end{align*}
after integration. A particular consequence is the estimate
\begin{align*}
\bnorm{\calN_\eps\ato{V}}_2\leq 
\eps^2CD \at{C+\eps^{3/2}D}^2
\end{align*}
for any $V\in B_D$, where we used that $\calN_\eps\ato{0}=0$.
\par 
\emph{\ul{Concluding arguments}}:
Combining all estimates derived so far with the definition of $\calF_\eps$ and the bounds for  $\calL_\eps^{-1}$ -- see Lemma \ref{Lem:InvertibilityOfLeps} -- we verify
\begin{align*}
\norm{\calF_\eps\ato{V}}_2\leq C + \eps^{3/2}CD^2 +
\eps^2CD \at{C+\eps^{3/2}D}^2
\end{align*}
for all $V\in B_D$ as well as
\begin{align*}
\norm{\calF_\eps\ato{V_2}}_2-\norm{\calF_\eps\ato{V_1}}_2\leq\at{\eps^{3/2}CD+
\eps^2C \at{C+\eps^{3/2}D}^2}\norm{V_2-V_1}_2
\end{align*}
for all $V_1,V_2\in B_D$. To complete the proof we first set $D:=2\,C$ and choose afterwards $\eps>0$ sufficiently small.
\end{proof}
\begin{corollary}[main result from \S\ref{sect:intro}]
\label{Cor:Summary}
For any sufficiently small $\eps>0$, 
the reformulated traveling wave equation \eqref{eq:scaledfpu1a} admits a unique even solution $W_\eps$ with speed $\sqrt{c_0^2+\eps^2}$ such that
\begin{align*}
\norm{W_\eps-W_0}_{2}+\norm{W_\eps-W_0}_{\infty}\leq C\eps^2
\end{align*}
holds for some constant $C$ independent of $\eps$. Moreover, $W_\eps$ is nonnegative and smooth.
\end{corollary}
\begin{proof}
The existence and local uniqueness of $W_\eps=W_0+\eps^2 V_\eps$ along with the $\fspaceL^2$-estimate is a direct consequence of Theorem \ref{Thm:FixedPoints}. Moreover, re-inspecting the arguments from the proof of Theorem \ref{Thm:FixedPoints} and using Lemma \ref{Lem:vonNeumann} we easily derive an uniform $\fspaceL^\infty$-bound for the corrector $V_\eps$. Finally, the right hand side in \eqref{Eqn:RescaledTWEqn} is -- at least for sufficiently small $\eps>0$ -- nonnegative due to the properties of the KdV wave $W_0$ and the potential $\Phi$, see Lemma \ref{Lem:LeadingOrder} and Assumption \ref{MainAssumption}. The nonnegativity of $W_\eps$ is hence granted by Corollary~\ref{Cor:InvarianceProperties}.
\end{proof}
The constants in the proof of Theorem \ref{Thm:FixedPoints} are, of course, far from being optimal. In general, a solution branch $\eps\to W_\eps\in\fspaceL^2_\even\at\Rset$ on an interval $\ccinterval{0}{\eps_*}$ can be continued for $\eps>\eps_*$  as long as the linearization of the traveling wave equation around $W_{\eps_*}$ provides an operator $\calL_{\eps_*}$ that can be inverted on the space $\fspaceL^2_\even\at\Rset$. Since the shift symmetry always implies that $W_{\eps_*}^\prime$ is an odd  kernel function of $\calL_{\eps_*}$, the unique continuation can hence only fail if the eigenvalue $c_{\eps_*}^2$ of the linearized traveling wave operator
\begin{align}
\label{Eqn:LinTWOp}
V\mapsto \sum_{m=1}^M m^2 A_{m\eps_*}\Phi_m^{\prime\prime}\at{m\eps_*^2 A_{m\eps_*} W_{\eps_*}}A_{m\eps_*} V
\end{align}
is not simple anymore. Unfortunately, almost nothing is known about the spectral properties of the operator \eqref{Eqn:LinTWOp} for moderate values $\eps_*$. It remains a challenging task to close this gap, especially since any result in this direction should have implications concerning the orbital stability of $W_{\eps_*}$.
\par
For $M=1$ it has also been shown in \cite[Propositions 5.5 and 7.1]{FP99} that the distance profile $\calA_\eps W_\eps$ is unimodal (`monotonic falloff') and decays exponentially for $x\to\pm\infty$. For $M>1$, it should be possible to apply a similar analysis to the velocity profile $W_\eps$ but the technical details are much more involved. It remains open to identify alternative and more robust proof strategies. For instance, if one could show that the waves from Corollary \ref{Cor:Summary} can be constructed by some variant of the abstract iteration scheme 
\begin{align*}
W\mapsto \calB_\eps^{-1}\at{Q_\eps\ato{W}+\eps^2\calP_\eps\ato{W}}\,,
\end{align*}
the unimodality of $W_\eps$ would be implied by the invariance properties of $\calA_{m\eps}$ and $\calB_\eps^{-1}$, see Lemma \ref{Lem:PropertiesOperatorA} and Corollary \ref{Cor:InvarianceProperties}. A similar argument could be used for the exponential decay because $\calA_{m\eps}$ maps a function with decay rate $\la$ to a function that decays with rate
\begin{align*}
\bar\la = \frac{\sinh\at{\tfrac12 \eps m \la }}{\tfrac12 \eps m \la}
\end{align*}
and since the von Neumann formula from Lemma \ref{Lem:vonNeumann} provides corresponding  expressions for $\calB_\eps^{-1}$; see \cite{HR10} for a similar argument to identify the decay rates of front-like traveling waves. In this context we further emphasize that only supersonic waves can be expected to decay exponentially. For subsonic waves with speed $c_\eps^2<c_0^2$, the linearization of the traveling wave equation \eqref{eq:scaledfpu1} predicts tails oscillations and hence non-decaying waves, see \cite{HMSZ13} for a similar analysis with non-convex interaction potentials.
%
%
\appendix
\addcontentsline{toc}{section}{List of symbols}
\section*{List of symbols}
\begin{tabular}{lcllllllll}
&$\al_m$, $\be_m$& linear and quadratic coefficients in force terms&& \eqref{eqn:ForceTerms}\\
& $\ga_m$& bounds for higher order force terms&& \eqref{MainAssumption.Eqn1}\\
&$c_\eps$& speed of the wave&& \eqref{eq:speedscaling}\\
\\ %
&$U_\eps$& position profile&&\eqref{eqn:TW.ansatz}\\
&$W_\eps$& velocity profile&&\eqref{Eqn:Def.Vel.prof}\\
&$V_\eps$&  corrector to the velocity profile&& \eqref{eqn.def.corrector}\\ %
&$R_\eps$, $S_\eps$&  residual terms with respect to $W_0$&&
\eqref{Lem.epsResidual.Eqn1}\\ %
\\ %
&$\calA_\eta$, $a_\eta$& convolution operator and its symbol function&& \eqref{eq:intoperator}, \eqref{Eqn:SymbolFct}\\
&$\calB_\eps$, $b_\eps$& auxiliary operator and its symbol function && \eqref{Eqn:OperatorBeps},  \eqref{Eqn:OperatorBSymb}\\
&$\Pi_\eps$, $\pi_\eps$& cut-off in Fourier space&&\eqref{eqn.cutoff.def} \\
&$\calQ_\eps$& quadratic terms in $W_\eps$&&\eqref{Eqn:NonlinOp.W} \\
&$\calP_\eps$& cubic and higher order terms in $W_\eps$&&\eqref{Eqn:NonlinOp.W} \\
&$\calM_\eps$& quadratic combination of $W_0$ and $V_\eps$&&\eqref{Eqn:DefLandM} \\
&$\calL_\eps$& linear terms in $V_\eps$&&\eqref{Eqn:DefLandM}\\
&$\calN_\eps$& remainder terms in   $V_\eps$&&\eqref{Eqn:NonlinOP.V} \\
&$\calF_\eps$& fixed point operator for $V_\eps$&& \eqref{Thm:FixedPoints.Eqn1} \\
\\
&$W_0$, $c_0$& velocity profile and speed of the KdV wave&&Lemma \ref{Lem:LeadingOrder}
\\
&$\calB_0$, $\calQ_0$& formal limits of 
$\calB_\eps$,  $\calQ_\eps$&&
\eqref{Eqn:OperatorB0},  \eqref{Eqn:NonlinOp.Q}\\
&$\calL_0$, $\calM_0$& formal limits of 
$\calL_\eps$, $\calM_\eps$&&
\eqref{Lem:LeadingOrder.Eqn1}
\end{tabular}
\addcontentsline{toc}{section}{Acknowledgements}
\section*{Acknowledgements}
The authors are grateful for the support by the Deutsche Forschungsgemeinschaft (DFG individual grant HE 6853/2-1) and the Austrian Science Fund (FWF grant J3143).
%

\end{document}